\documentclass[aos,preprint]{imsart}
\usepackage{amsmath,amssymb}                
\usepackage{graphicx}
\usepackage{amsthm}
\usepackage{color}
\newtheorem{proposition}{Proposition}[section]
\newtheorem{corollary}{Corollary}[section]
\newtheorem{lemma}{Lemma}[section]
\newtheorem{remark}{Remark}[section]
\newtheorem{theorem}{Theorem}[section]

\newcommand{\mP}{\mathbb P}
\newcommand{\bU}{\mathbf U}
\newcommand{\bu}{\mathbf u}
\newcommand{\U}{\mathcal U}
\renewcommand{\S}{\mathfrak S}
\newcommand{\FWER}{\mathrm{FWER}}

\newcommand{\false}{\mathrm{false}}

\newcommand{\tv}{\theta_V}
\newcommand{\tg}{\theta_G}
\newcommand{\eg}{\varepsilon_G}
\newcommand{\bxi}{\boldsymbol{\xi}}
\newcommand{\floor}[1]{[#1]}
\newcommand{\cov}{\mathrm{cov}}
\renewcommand{\bullet}{\odot}

\begin{document}
\begin{frontmatter}
\title{Coarse-to-fine Multiple Testing Strategies\thanksref{T1}}
\runtitle{Coarse-to-fine Multiple Testing Strategies}
\thankstext{T1}{This work was partially supported by NSF:1228248.}
\begin{aug}
\author{\fnms{Kamel} \snm{Lahouel}\ead[label=e1]{klahoue1@jhu.edu}},
\author{Donald Geman\ead[label=e2]{geman@jhu.edu}}
\and
\author{Laurent Younes\ead[label=e3]{laurent.younes@jhu.edu}}
\runauthor{Kamel Lahouel et al.}
\affiliation{Center for Imaging Science and Department of Applied Mathematics and Statistics, Johns Hopkins University}
\address{Center for Imaging Science\\
Johns Hopkins University\\
3400 N. Charles st.\\
Baltimore MD 21218, USA\\
\printead{e1}\\
\printead{e2}\\
\printead{e3}}
\end{aug}
\begin{abstract}
We analyze control of the familywise error rate (FWER) in a multiple
testing scenario with a great many null hypotheses about the
distribution of a high-dimensional random variable among which only a
very small fraction are false, or ``active''.  In order to improve
power relative to conservative Bonferroni bounds, we explore a
coarse-to-fine procedure adapted to a situation in which tests are
partitioned into subsets, or ``cells'', and active hypotheses tend to
cluster within cells.  We develop procedures for a standard linear
model with Gaussian data and a non-parametric case based on
generalized permutation testing, and demonstrate considerably higher
power than Bonferroni estimates at the same FWER when the active
hypotheses do cluster. The main technical difficulty arises from the
correlation between the test statistics at the individual and cell
levels, which increases the likelihood of a hypothesis being falsely
discovered when the cell that contains it is falsely discovered
(survivorship bias).  This requires sharp estimates of certain
quadrant probabilities when a cell is inactive.
\end{abstract}


\end{frontmatter}


 \section{Introduction}
 \label{sec:1}

We consider a multiple testing scenario encountered in many current
applications of statistics.  Given a large index set $V$ and a family $(H_0(v), v\in V)$ of null
hypotheses about the distribution of a high-dimensional random vector
$U \in \mathbb R^d$, we wish to design a procedure, basically a family
of test statistics and thresholds, to estimate the subset $A \subset
V$ over which  the null hypotheses are false.  We shall refer to $A$
as the ``active set'' and write $\hat A = \hat A({\bf U})$ for our
estimator of $A$ based on a random sample $\bf U$ of size $n$ from
$U$. The hypotheses in $\hat A({\bf U})$ (namely the ones for which
the null is rejected) are referred to as ``detections'' or
``discoveries.''  Naturally, the goal is to maximize the number $|A
\cap \hat A({\bf U})|$ of detected true positives while simultaneously
controlling the number $|A^c \cap \hat A({\bf U})|$ of false discoveries.

There are two widely used criteria for controlling false
positives:\\

{\bf FWER:} Assume that $\mathbf U$ is defined on the probability
space $(\Omega, \mathbb P)$. The family-wise error rate (FWER) is
$$\FWER(\hat A)=\mathbb{P}\left(\hat{A}({\bf U}) \cap A^c \neq
\emptyset \right),$$ which is the probability of making at least one
false discovery. This is usually controlled using Bonferroni bounds
and their refinements
\cite{gabriel1969simultaneous,Holm1979,hochberg1988sharper,good2013permutation},
or using resampling methods or random permutation.\\

{\bf FDR:} The false discovery rate (FDR) is the expected ratio
between the number of false alarms $|A^c \cap \hat A({\bf U})|$ and
the number of discoveries $|\hat A({\bf U})|$
\cite{benjamini1995controlling,benjamini2001control,benjamini2006adaptive}.\\

In many cases, including the settings in computational biology which
directly motivate this work, we find $|A| \ll |V|$, $n \ll d$ as well
as small ``effect
sizes.''  This is the case, for example, in genome-wide
association studies (GWAS) where $U=(Y,X_v, v \in V)$
and the dependence of the ``phenotype'' $Y$ on the ``genotype'' $(X_v, v \in
V)$ is often assumed to be linear; the active set $A$ are those $v$ with
non-zero coefficients and effect size refers to the fraction of the
total variance of $Y$ explained by a particular $X_v$. Under these
challenging circumstances, the FWER criterion is usually very
conservative and power is limited; that is, number of true positive
detections is often very small (if not null) compared to $|A|$ (the
``missing heritability'').  This is why the less conservative FDR
criterion is sometimes preferred: it allows for a higher number of
true detections, but of course at the expense of false
positives. However, there are situations, such as GWAS, in which this
tradeoff is unacceptable; for example, collecting more data and doing
follow-up experiments may be too labor intensive or expensive, and
therefore having even one false discovery may be deemed undesirable.

To set the stage for our proposal, suppose we are given a family
$T_v=T_v({\bf U}),v \in V$ of test statistics and can assume that
deviations from the null are captured by small values of $T_v({\bf
U})$ (e.g., p-values).  We make the usual assumption, easily achieved
in practice, that the distribution of $T_v({\bf U})$ does not depend
on $v$ when $v\in A^c$, and individual rejection regions are of the
form $\{u\in \U: T_v(u) \leq
\theta \}$ for a constant $\theta$ independent of $v$.  Defining
$\hat A({\bf U}) = \{v: T_v({\bf U}) \leq \theta \}$, the 
Bonferroni upper-bound is
\[
\FWER \leq \sum_{v\in A^c} \mP(T_v({\bf U}) \leq \theta) 
\leq |V| \max_{v\in A^c} \mP(T_v({\bf U}) \leq \theta).\]
To ensure that $\FWER \leq \alpha$, $\theta=\theta_B$ is selected such
that $\mP(T_v({\bf U}) \leq \theta_B) \leq \alpha/|V|$ whenever $v \in
A^c$. The Bonferroni bound can only be marginally improved (see, in
particular estimator \cite{Holm1979}, which will be referred to as Bonferroni-Holm in the rest of the paper) in the general case. While
alternative procedures (including permutation tests) can be designed
to take advantage of correlations among tests, the bound is sharp when
$|V| \gg |A|$ and tests are independent.

\medskip

\noindent
{\bf Coarse-to-fine Testing:} Clearly some additional assumptions or
domain-specific knowledge is necessary to ameliorate the reduction in
power resulting from controlling the FWER.  Motivated by applications
in genomics, we suppose the set $V$ has a natural hierarchical
structure.  In principle, it should then be possible to gain power if the
active hypotheses are not randomly distributed throughout $V$ but
rather have a tendency to cluster within cells of the hierarchy.  In
fact, we shall consider the simplest example consisting of only two
levels corresponding to individual hypotheses indexed by $v \in V$ and a
partition of $V$ into non-overlapping subsets $(g
\subset V, g \in G)$, which we call ``cells.'' 
We will propose a particular multiple testing strategy which is
coarse-to-fine with respect to this structure, controls the
FWER, and whose power will exceed that of the standard Bonferroni-Holm
approach for typical models and realistic parameters when a minimal
degree of clustering is present.  
It is important to note that
clustering property is not a condition for a correct control of the
FWER at a given level using our coarse-to-fine procedure, but only for its
increased efficiency in discovering active hypotheses.

Our estimate of $A$ is now based on two families of test statistics:
$\{T_v({\bf U}), v \in V\}$, as above, and $\{T_g({\bf U}), g \in G\}$.  
The cell-level test $T_g$ is designed to assume
small values only when $g$ is ``active,'' meaning that $g \cap A \neq \emptyset$.
Our estimator of $A$ is now
\[  \hat A({\bf U}) = \{v: T_g({\bf U}) \leq \tg,\,T_v({\bf U}) \leq \tv \}.\]
One theoretical challenge of this method is to derive a tractable
method for controlling the FWER at a given level $\alpha$.  Evidently,
this method can only out-perform Bonferroni if
$\tv > \theta_B$; otherwise, the coarse-to-fine active set is a subset of the
Bonferroni discoveries. A key
parameter is $J$, an upper bound on the number of active cells, and in the next
section we will derive an FWER bound 
\[ \FWER(\hat A({\bf U})) \leq \Phi(\tg,\tv,J) \]
under an appropriate compound null hypothesis.

The main results of the paper are in the ensuing analysis for
different models for $U$.  In each case, the first objective is to
compute $\Phi$ for a given $\tg$ and $\tv$ and the second objective is
to maximize the power over all pairs $(\tg,\tv)$ which
satisfy $\Phi \leq \alpha$.  The smaller our upper bound on $J$, the
stronger is the clustering of active hypotheses in cells and the
greater is the gain in power compared with the Bonferroni bound.  In
particular, as soon as $J\ll|G|$, the coarse-to-fine strategy will lead to a
considerably less conservative score threshold for individual
hypotheses relative to the Bonferroni estimate and the coarse-to-fine procedure
will yield an increase in power for a given FWER.  Again, our
assumptions about clustering are only expressed through an upper bound
on $J$; no other assumptions about the distribution of $A$ are made
and the FWER is controlled in all cases.

The main technical difficulty arises from the correlation between the
corresponding test statistics.  This must be taken into account since
it increases the likelihood of an individual index $v$ being falsely
declared active when the cell $g(v)$ that contains it is falsely discovered
(survivorship bias).  More specifically, we require sharp estimates of
quadrant probabilities under the {\it joint distribution} of
$T_{g(v)}({\bf U})$ and $T_v({\bf U})$ when 
$g(v)$, the cell containing $v$, is inactive.  All these issues will
be analyzed in two cases.  First, we will consider the standard linear
model with Gaussian data.  In this case $\Phi$ is
expressed in terms of centered chi-square distributions and the power
is expressed in terms of non-centered chi-square distributions. The efficiency of
the coarse-to-fine method in detecting active hypotheses will depend on effect
sizes, both at the level of cells and individual $v$, among other
factors.  A non-parametric procedure will then be developed in section
\ref{sec:5} based on generalized permutation testing and invariance assumptions.
Finally, we shall derive a high-confidence upper bound on $J$ based on
a martingale argument.  Extensive simulations comparing the power of
the coarse-to-fine and Bonferroni-Holm appear throughout.

\medskip

\noindent
{\bf Applications and Related Work:} As indicated above, our work (and
some of our notation) is inspired by statistical issues
arising in GWAS \cite{corvin2010genome,fridley2009bayesian,balding2006tutorial} 
and related areas in computational genomics.  In the
most common version of GWAS, the ``genotype'' of an individual is
represented by the genetic states $X_v$ at a very large
family of genomic locations $v \in V$; these variations are called
single nucleotide polymorphisms or SNPs.  In any given study the
objective is to find those SNPs $A \subset V$ ``associated'' with a
given ``phenotype'', for example a measurable trait $Y$ such as height
or blood pressure.  The null hypothesis for SNP $v$ is that $Y$ and
$X_v$ are independent r.v.s, and whereas $|V|$ may run into the
millions, the set $A$ of active variants is expected to be fewer than
one hundred. (Ideally, one seeks the ``causal'' variants, an even
smaller set, but separating correlation and causality is notoriously
difficult.)  Control of the FWER is the gold standard and the linear
model is common.  If the considered variants are confined to coding
regions, then the set of genes provides a natural partition of $V$
(and the fact that genes are organized into pathways provides a
natural three-level hierarchy) \cite{huang2011gene}

Another application of large-scale multiple testing is variable
filtering in high-dimensional prediction: the objective is to predict
a categorical or continuous variable $Y$ based on a family of
potentially discriminating features $X_v, v \in V$.  Learning a
predictor $\hat{Y}$ from i.i.d. samples of $U=(Y,X_v,v \in V)$ is
often facilitated by limiting {\it a priori} the set of features utilized 
in training $\hat{Y}$  to
a subset $A \subset V$ determined by 
testing the features one-by-one for dependence on $Y$
and setting a signficance threshold. In most
applications of machine learning to artificial perception, no premium
is placed on pruning $A$ to a highly distinguished subset; indeed, the
particular set of selected features is rarely examined or considered
of significance.  In contrast, the identities of the particular features
selected and appearing in decision rules are often of keen interest in
computational genomics, e.g., discovering cancer biomarkers, where the
variables $X_v$ represent ``omics'' data (e.g., gene expression), and
$Y$ codes for two possible cellular or disease phenotypes.  Obtaining
a ``signature'' $\hat{A}$ devoid of false positives can be beneficial
in understanding the underlying biology and interpreting the decision
rules.  In this case the Gene Ontology (GO) 
\cite{ashburner2000gene}  provides a very rich
hierarchical structure, but one example being the organization of
genes in pathways. Indeed, building predictors to separate ``driver
mutations'' from ``passenger mutations'' in cancer would appear to be
a promising candidate for coarse-to-fine testing due to the fact that drivers are
known to cluster in pathways.

There is a literature on coarse-to-fine pattern recognition (see, e.g., \cite{blanchard2005hierarchical}
and the references therein), but the emphasis has traditionally been on computational efficiency
rather than error control. Computation is not considered here.
Moreover, in most of this work, especially applications to vision and
speech, the emphasis is on detecting true positives (e.g., patterns of
interest such as faces) at the expense of false positives. Simply
``reversing'' the role of true positives and negatives is not feasible
due to the loss of reasonable invariance assumptions; in effect,
every pattern of interest is unique.

Finally, in \cite{meinshausen2008hierarchical}, a hierarchical testing approach is
used in the context of the FWER.  However, the intention is to improve
the power of detection relative to the Bonferroni-Holm methods only at
level of clusters of hypotheses; in contrast to our method, the two
approaches have comparable power at the level of individual
hypotheses.

\medskip

\noindent
{\bf Organization of the Paper:} The paper is structured as follows:
In section \ref{sec:2} we present a Bonferroni-based inequality that
will be central for controlling the FWER using the coarse-to-fine
method in different models. In section \ref{sec:3} will consider a
parametric model that will illustrate precisely the way we control the
FWER at a fixed level and permit a power comparison between coarse-to-fine and
Bonferroni-Holm.  We then propose a non-parametric procedure in section
\ref{sec:5} under general invariance assumptions. 
A method for estimating an upper bound on the
number of active cells and incorporating it into the testing procedure
without violating the FWER constraint is derived in section \ref{sec:6}.
Finally, some concluding remarks are made in the Discussion.


\section{Coarse-to-fine framework}
\label{sec:2}
 
The finite family of null hypotheses will be denoted by $(H_0(v),
v\in V)$, where $H_0$ is either true or false.  We are interested in
the {active set}  of indices, $A = \{v\in V: H_0(v) = \text{false}\}$ and
will write $V_0 = A^c$ for the set of inactive indices.  Suppose
our data $\bf U$ takes values in $\U$.  The set $\hat A({\bf U})$ is
commonly designed based on individual rejection regions
$\Gamma_v\subset \U$, with $\hat A({\bf U}) = \{v: {\bf U} \in
\Gamma_v\}.$ As indicated in the previous section, in the conservative
Bonferroni approach, the $\FWER$ is controlled at level $\alpha$ by
assuming $|V| \max_{v\in V_0} \mP({\bf U} \in \Gamma_v) \leq
\alpha$.  If the rejection regions are designed so that this probability
is independent of $v$ whenever $H_0(v) = \mathrm{true}$, then the condition
boils down to $\mP({\bf U} \in \Gamma_v) \leq \alpha/|V|$ for $v \in V_0$.
Generally, $\Gamma_v = \{u \in \U: T_v(u) \leq t\}$ for a constant $t$
for some family of test statistics $(T_v, v \in V)$.

While there is not much to do in the general case to improve on the
Bonferroni method, it is possible to improve power if $V$ is
structured and one has prior knowledge about way the active hypotheses
are organized relative to this structure.
In this paper, we consider a
coarse-to-fine framework in which $V$ is provided with a
partition $G$, so that $V = \bigcup_{g\in G}g$, where the subsets
$g\subset V$ (which we will call { cells}) are non-overlapping. For
$v\in V$, we let $g(v)$ denote the unique cell $g$ that contains
it. The ``coarse'' step selects cells
likely to contain active indices, followed by a ``fine'' step in which a
Bonferroni or equivalent procedure is applied only to hypotheses included in
the selected cells. More explicitly, we will associate a rejection
region $\Gamma_g$ to each $g\in G$ and consider the
discovery set 
\begin{equation} \label{eq:A.ctf} \hat A({\bf U}) = \{v\in V:
{\bf U} \in \Gamma_{g(v)}\cap\Gamma_v\}.  
\end{equation}
We will say that a cell $g$ is active if and only if $g\cap A \neq
\emptyset$, which we shall also express as $H_0(g) = \false$,
implicitly defining $H_0(g)$ as the logical ``and" of all $H_0(v), v\in
g$. We will also consider the double null hypothesis $H_{00}(v) =
H_0(g(v))$ of $v$ belonging in an inactive cell (which obviously
implies that $v$ is inactive too), and we will let $V_{00} \subset
V_0$ be the set of such $v$'s.

Let $\nu_G$ denote the size of the largest cell in $G$ and $J$
be the number of active cells.
We will develop our procedure under the assumption that $J$ 
is known, or, at least bounded from above. While this can
actually be a plausible assumption in practice, we will relax it in
section \ref{sec:5} in which we will design a procedure to 
estimate a bound on $J$.
  Then under these assumptions we have the following result:
\begin{proposition}
\label{prop:first.bound}
With $\hat A$ defined by \eqref{eq:A.ctf}:
$$ 
\FWER(\hat A) \leq |V|\, \max_{v\in V_{00}} \mathbb{P}\left({\bf U} \in \Gamma_{g(v)} \cap \Gamma_v\right)+\nu_G \, J \,  \max_{v\in V_0}\mathbb{P}\left({\bf U} \in \Gamma_v \right).$$
\end{proposition}
\begin{proof}
This is just the Bonferroni bound applied to the decomposition
\begin{eqnarray*}
\hat A({\bf U}) \cap V_0 \neq \emptyset&=& \bigcup_{v\in V_{00}}({\bf U}\in \Gamma_{g(v)} \cap \Gamma_v) \cup \bigcup_{v\in V_0\setminus V_{00}}({\bf U}\in \Gamma_{g(v)} \cap \Gamma_v)\\
&\subset& \bigcup_{v\in V_{00}}({\bf U}\in \Gamma_{g(v)} \cap \Gamma_v) \cup \bigcup_{v\in V_0\setminus V_{00}}({\bf U}\in \Gamma_v)
\end{eqnarray*}
so that
\begin{eqnarray*}
P(\hat A({\bf U}) \cap V_0\neq \emptyset) \leq |V_{00}|\max_{v\in V_{00}} \mathbb{P}\left({\bf U} \in \Gamma_{g(v)} \cap \Gamma_v\right) + |V_0\setminus V_{00}| \max_{v\in V_0} \mP({\bf U}\in \Gamma_v)
\end{eqnarray*}
and the proposition results from $|V_{00}| \leq |V|$ and $|V_0 \setminus V_{00}| \leq \nu_G\, J$.
\end{proof}

The sets $\Gamma_g$ and
$\Gamma_v$ will be designed using statistics $T_g({\bf U})$ and
$T_v({\bf U})$ setting $\Gamma_g = [T_g({\bf U}) \leq
\tg]$ and $\Gamma_v = [T_v({\bf U}) \leq \tv]$ for some constants
$\tg$ and $\tv$, and assuming that the distribution of $(T_{g(v)}({\bf U}),
T_v({\bf U}))$ (resp. $T_v({\bf U})$) is independent of $v$ for $v\in
V_{00}$ (resp. $v\in V_0$).   Letting $p_{00}(\tg, \tv) = \mathbb{P}
\left(\{T_{g(v)}({\bf U}) \leq \tg\} 
\cap \{T_v({\bf U}) \leq \tv\}\right)$ for $v\in V_{00}$ and
$p_{0}(\tv) = \mathbb{P}\left(T_v({\bf U}) \leq \tv\right)$ for $v\in
V_0$, the previous upper bound becomes
\begin{equation}
\label{eq:first.bound}
\FWER(\hat A) \leq |V|\, p_{00}(\tg, \tv) +\nu_G \, J\,  p_0(\tv).
\end{equation}

In the following sections our goal will be to design $\tg$ and $\tv$
such that this upper bound is smaller than a predetermined level
$\alpha$. Controlling the second term will lead to less conservative
choices of the constant $\tv$ (compared to the Bonferroni estimate),
as soon as $\nu_G J \ll |V|$ (or $J\ll|G|$ if all cells have
comparable sizes).  Depending on the degree of clustering, the
probability $p_{00}$ of false detection in the two-step procedure can
be made much smaller than $p_{0}$ without harming the true detection
rate and the coarse-to-fine procedure will yield an increase in power for a given
FWER. We require tight estimates of $p_{00}$ and taking into account
the correlation between $T_{g(v)}({\bf U})$ and $T_v({\bf U})$ is
necessary to deal with ``survivorship bias.''


\section{Model-based derivation}
\label{sec:3}
\subsection{Regression model}
In this section, the observation is a realization  of an  i.i.d. family  of random  variables $U = ((Y^k,X^k), k=1, \ldots n)$ where  the ${Y}$'s  are  real-valued and the variables $X^k = (X^k_v, v\in V)$ is a high-dimensional family of variables indexed by  the set $V$. We assume that the distribution of $X^k_v, v\in V$, are independent and  centered Gaussian, with variance $\sigma_v^2$ , and that
$$Y^k=a_0 + \sum_{v \in A} a_v X_v^k+\xi^k$$
where $\xi^1, \ldots, \xi^n$ are i.i.d. Gaussian with variance $\sigma^2$ and  $a_v, v\in A$, are unknown real coefficients.  We will denote by $\mathbf{Y}$ the vector $(Y^1,\ldots,Y^n)$ and by $\bar{\mathbf{Y}} = \left(\sum_{k=1}^{n}{Y^k}/n\right)\boldsymbol 1_n$ where  $\mathbf{1}_n$ is the vector composed by ones repeated $n$ times. We also let $\mathbf{X}_v=(X_v^1,\ldots,X_v^n)$ and $\xi = (\xi^1,\ldots,\xi^n)$, so that 
$$\mathbf{Y}=\sum_{v \in A} a_v {\mathbf{X}_v}+\xi.$$

Finally, we will denote by $\sigma_\mathbf{Y}^2$  the common variance of $Y^1, \ldots, Y^n$ and assume that it is known (or estimated from the observed data).

\subsection{Scores}
For $v\in V$, we denote  by $P_v$ the orthogonal projection on the subspace  $S_v$ spanned by the two vectors $\mathbf{X}_v$ and $\mathbf{1}_n$. We will also  denote by $P_g$ ($g\in G$) the orthogonal projection on the subspace $S_g$  spanned by the vectors $\mathbf{X}_v$, $v \in g$, and $\mathbf{1}_n$. The scores at the $g$ level and $v$ level will  be respectively:
$$T_g(U) = \frac{{\lVert P_g\mathbf{Y} \rVert}^2-{\lVert \bar{\mathbf{Y}} \rVert}^2}{\sigma_\mathbf{Y}^2}$$
and
$$T_v(U) = \frac{{\lVert P_v\mathbf{Y} \rVert}^2-{\lVert \bar{\mathbf{Y}} \rVert}^2}{\sigma_\mathbf{Y}^2}.$$
(The projections are simply obtained by least-square regression of $\mathbf{Y}$ on  $\mathbf{X}_v, v \in g$, for $P_g$ and  on $\mathbf{X}_v$ for $P_v$.) We now provide estimates of 
$$p_{00}(\tg,\tv)=\mathbb{P}\left(\frac{{\lVert P_g\mathbf{Y} \rVert}^2-{\lVert \bar{\mathbf{Y}} \rVert}^2}{\sigma_\mathbf{Y}^2}>\tg; \frac{{\lVert P_v\mathbf{Y} \rVert}^2-{\lVert \bar{\mathbf{Y}} \rVert}^2}{\sigma_\mathbf{Y}^2}>\tv\right)$$
for $v\in V_{00}$ and $g=g(v)$ and
$$p_{0}(\tv)=\mathbb{P}\left(\frac{{\lVert P_v\mathbf{Y} \rVert}^2-{\lVert \bar{\mathbf{Y}} \rVert}^2}{\sigma_\mathbf{Y}^2}>\tv\right)$$
for $v\in V_0$. Note that, because we consider residual sums of squares, we here use large values of the scores in the rejection regions (instead of small values in the introduction and other parts of the paper), hopefully without risk of confusion.

\begin{proposition}
\label{prop:1}
For  all $\tg$ and $\tv$:
\begin{multline*}
 p_{00}(\tg,\tv)\leq C(\nu_G)\exp\left(-\frac{\theta_G}{2}\right)\theta_G^{\frac{\nu_G}{2}}\left(1-G_{\beta}\left(\frac{\theta_V}{\theta_G},\frac{1}{2},\frac{\nu_G+1}{2}\right)\right)
 \\+\left(1-F_1(\tg-\nu_G+1)\right),
 \end{multline*}
where $G_{\beta}(x,a,b)$ is the CDF of a $\beta(a,b)$ distribution evaluated at $x$ and:
\[
C(\nu_G)=\frac{\exp{(\frac{\nu_G-1}{2})}}{\sqrt{2}{(\nu_G-1)}^{(\frac{\nu_G-1}{2})}}\frac{\Gamma(\frac{\nu_G}{2}+\frac{1}{2})}{\Gamma(\frac{\nu_G}{2}+1)}.\]
Moreover
\[p_{0}(\tv)\leq 1-F_1(\tv)
\]
where $F_k$ is the c.d.f. of a chi-squared distribution with $k$ degrees of freedom.

\end{proposition}

\begin{proof}
For $v\in V_{00}$ and $g = g(v)$, we can write
\begin{multline*}
\mathbb{P}\left(\frac{{\lVert P_g\mathbf{Y} \rVert}^2-{\lVert \bar{\mathbf{Y}} \rVert}^2}{\sigma_\mathbf{Y}^2}>\tg; \frac{{\lVert P_v\mathbf{Y} \rVert}^2-{\lVert \bar{\mathbf{Y}} \rVert}^2}{\sigma_\mathbf{Y}^2}>\tv\right) \\
= \mathbb{P}\left(\frac{{\lVert P_g\mathbf{Y} \rVert}^2-{\lVert \bar{\mathbf{Y}} \rVert}^2}{\sigma_{\mathbf{Y}-g}^2}>\tg; \frac{{\lVert P_v\mathbf{Y} \rVert}^2-{\lVert \bar{\mathbf{Y}} \rVert}^2}{\sigma_{\mathbf{Y}-g}^2}>\tv\right)$$
\end{multline*}
because  $\sigma_{\mathbf{Y}-g}^{2}=\sum_{v \in A \cap g^c}a_v^2 \sigma_v^2+\sigma^2$ and $A\cap g^c = A$.

Consider the conditional probability:
$$\mathbb{P}\left(\frac{{\lVert P_g\mathbf{Y} \rVert}^2-{\lVert \bar{\mathbf{Y}} \rVert}^2}{\sigma_{\mathbf{Y}-g}^2}>\tg; \frac{{\lVert P_v\mathbf{Y} \rVert}^2-{\lVert \bar{\mathbf{Y}} \rVert}^2}{\sigma_{\mathbf{Y}-g}^2}>\tv\,\Big|\,(\mathbf{X}_v)_{v \in g} \right).$$

The conditional distribution of $\mathbf{Y}$ given  $(\mathbf{X}_v)_{v \in g}$ is Gaussian $\mathcal{N}(0,\sigma_{\mathbf{Y}-g}^2  \times I_n)$ (where $I_n$ is the $n$-dimensional identity matrix).  Denote by $P_v'$ the projection on the orthogonal complement of $\mathbf{J}$ in $S_v$ and by $P_g'$ the projection on the orthogonal complement of $S_{v}$ in $S_g$, so that
$${\lVert P_g\mathbf{Y} \rVert}^2-{\lVert \bar{\mathbf{Y}} \rVert}^2={\lVert P_{g}'\mathbf{Y} \rVert}^2+{\lVert P_{v}'\mathbf{Y} \rVert}^2$$
and
$${\lVert P_v\mathbf{Y} \rVert}^2-{\lVert \bar{\mathbf{Y}} \rVert}^2={\lVert P_{v}'\mathbf{Y} \rVert}^2.$$
This implies that:
\begin{multline*}
\mathbb{P}\left(\frac{{\lVert P_g\mathbf{Y} \rVert}^2-{\lVert \bar{\mathbf{Y}} \rVert}^2}{\sigma_{\mathbf{Y}-g}^2}>\tg; \frac{{\lVert P_v\mathbf{Y} \rVert}^2-{\lVert \bar{\mathbf{Y}} \rVert}^2}{\sigma_{\mathbf{Y}-g}^2}>\tv\,\Big|\, (\mathbf{X}_v)_{v \in g}\right)
= \\
\mathbb{P}\left(\frac{{\lVert P_g'\mathbf{Y} \rVert}^2+{\lVert P'_v \mathbf{Y}\rVert}^2}{\sigma_{\mathbf{Y}-g}^2}>\tg; \frac{{\lVert P'_v\mathbf{Y} \rVert}^2}{\sigma_{\mathbf{Y}-g}^2}>\tv\,\Big|\, (\mathbf{X}_v)_{v \in g}\right)
\end{multline*}
At this stage,  applying Cochran's theorem to $P_{g}'(\mathbf{Y}/\sigma_{\mathbf{Y}-g})$ and $P_{v}'(\mathbf{Y}/\sigma_{\mathbf{Y}-g})$, which are conditionally independent given $\mathbf{X}_v, v\not\in G$,  reduces the problem to finding an upper bound for:
$$\mathbb{P}\left(\eta+\zeta \geq \theta_G;\zeta \geq \theta_V \right),$$
where $\eta$ is $\chi^2(\nu_G-1)$ and $\zeta$ is $\chi^2(1)$, and the two variables are independent. Let us write this probability as 
$$\mathbb{E}\left(\mathbf{1}_{\eta+\zeta \geq \theta_G} \mathbf{1}_{\zeta \geq \theta_V}  \mathbf{1}_{\zeta < \theta_G-\nu_G+1}\right)+\mathbb{E}\left(\mathbf{1}_{\eta+\zeta \geq \theta_G} \mathbf{1}_{\zeta \geq \theta_V}  \mathbf{1}_{\zeta \geq \theta_G-\nu_G+1}\right),$$
which is less than:
$$\mathbb{E}\left(\mathbf{1}_{\eta+\zeta \geq \theta_G} \mathbf{1}_{\zeta \geq \theta_V}  \mathbf{1}_{\zeta < \theta_G-\nu_G+1}\right)+(1-F_1(\theta_G-\nu_G+1)).$$
(Here, $\mathbb E$ refers to the expectation with respect to $\mathbb P$.)

Consider the first term in the sum: $\mathbb{E}\left(\mathbf{1}_{\eta+\zeta \geq \theta_G} \mathbf{1}_{\zeta \geq \theta_V}  \mathbf{1}_{\zeta < \theta_G-\nu_G+1}\right)$.
This term can be re-written as: 
$$\mathbb{E}\left(\mathbb{E}(\mathbf{1}_{\eta \geq \theta_G-\zeta}|\zeta) \mathbf{1}_{\zeta \geq \theta_V}  \mathbf{1}_{\zeta < \theta_G-\nu_G+1}\right).$$
At this stage, we will use the following tail inequality for $\chi^2(k)$ random variables :
$$1-F_k(zk) \leq {(z \exp(1-z))}^{\frac{k}{2}},$$
for  any $z>1$. We apply this result to $k=\nu_G-1$ and $z=\frac{\theta_G-V}{\nu_G-1}$ to get the upper bound:
$$\mathbb{E}\left(\mathbb{E}(\mathbf{1}_{\eta \geq \theta_G-\zeta}|\zeta) \mathbf{1}_{\zeta \geq \theta_V}  \mathbf{1}_{\zeta < \theta_G-\nu_G+1}\right) \leq  \mathbb{E}\left(\left(\frac{\theta_G-\zeta}{\nu_G-1}\exp\left(1-\frac{\theta_G-\zeta}{\nu_G-1}\right)\right)^{\frac{\nu_G-1}{2}} \mathbf{1}_{\zeta \geq \theta_V}  \right).$$
Since the density of a $\chi^2(1)$ is proportional to $\exp(-\frac{\zeta}{2})\zeta^{-\frac{1}{2}}$, the term in $\exp{\frac{\zeta}{2}}$ will cancel in the last integral (expectation). Using a simple change of variables in the remaining integral, we  have as a final upper bound:
$$C(\nu_G)\exp\left(-\frac{\theta_G}{2}\right)\theta_G^{\frac{\nu_G}{2}}\left(1-G_{\beta}\left(\frac{\theta_V}{\theta_G},\frac{1}{2},\frac{\nu_G+1}{2}\right)\right),$$
where $G_{\beta}(x,a,b)$ is the CDF of a Beta(a,b) evaluated at $x$.  \\
The  second upper-bound, for $p_0(\tv)$, is easily  obtained, the proof being left to the reader.
\end{proof}

 This leads us immediately to the following  corollary:
 \begin{corollary}
 \label{cor:2}
 With  the thresholds $\tg$ and $\tv$, an upper bound of the FWER is:
 \begin{multline}
 \label{eq:cor.2}
 \FWER(\hat A) \leq |V| C(\nu_G)\exp\left(-\frac{\theta_G}{2}\right)\theta_G^{\frac{\nu_G}{2}}\left(1-G_{\beta}\left(\frac{\theta_V}{\theta_G},\frac{1}{2},\frac{\nu_G+1}{2}\right)\right) \\
 + J \nu_G \left(1-F_1(\tv)\right).
\end{multline}
\end{corollary}

Figure \ref{fig:level.curves} provides an illustration of the level curves associated to the  above FWER upper bound. More precisely, it illustrates the tradeoff between the conservativeness at the cell level and the individual index level. In the next section, the optimization for power will be made along these level lines. Figure 1 also provides the value of  the Bonferroni-Holm threshold. For the coarse-to-fine procedure to be less conservative  than the Bonferroni-Holm approach, we need the index-level threshold to be smaller, i.e., the optimal point on the level line to be chosen below the corresponding dashed line.

\begin{figure}
\centering
\includegraphics [width=.95\textwidth]{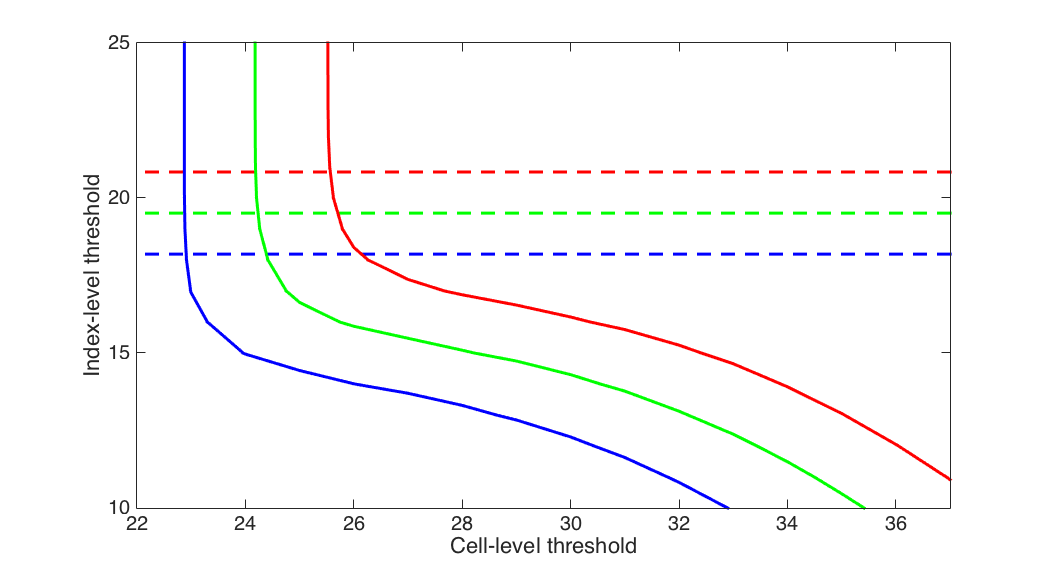}
\caption{\label{fig:level.curves} Level curves  of the upper bound of the FWER for the levels 0.2 (blue), 0.1 (green) and 0.05 (red). The horizontal dashed lines represent the thresholds at the individual level for a Bonferroni-Holm test, with  corresponding colors.}
\label{fig:di}
\end{figure} 
 
The derivation of \eqref{eq:cor.2} is based on the assumption that we have a fixed cell size (across all the cells), which is not needed. In the case where the size of the cell is varying, it is easy to generalize the previous upper bound. Letting
\[
\phi(\nu_G, \tg,\tv) = C(\nu_G)\exp\left(-\frac{\theta_G}{2}\right)\theta_G^{\frac{\nu_G}{2}}\left(1-G_{\beta}\left(\frac{\theta_V}{\theta_G},\frac{1}{2},\frac{\nu_G+1}{2}\right)\right)
\]
it suffices to  replace $|V| \phi(\nu_G, \tg,\tv) $ in \eqref{eq:cor.2} with  $\sum_{g\in G} |g| \phi(|g|, \sqrt{|g|}\tg,\tv)$ where $\tg$ does not depend on the cell $g$.

\subsection{Optimal thresholds}
Equation \eqref{eq:cor.2} provides a constraint on the pair $(\tg,\tv)$ to control  the FWER at a given level $\epsilon$. We now show how to obtain ``optimal'' thresholds $(\tg^*,\tv^*)$ that maximize discovery subject to this constraint. The discussion will also help understanding how active indices clustering in cells improve the power of the coarse-to-fine procedure.

The conditional distribution of $\mathbf{Y}$  given $(\mathbf{X}_v, v\in g)$ is $\mathcal N(\sum_{v\in g\cap A} a_v \mathbf{X}_v, \sigma^2_{\mathbf{Y}-g})$ with $\sigma_{\mathbf{Y}-g}^{2}=\sum_{v \in A \cap g^c}a_v^2 \sigma_v^2+\sigma^2$. 
It follows from this that, conditionally to these variables, $({\lVert P_g\mathbf{Y} \rVert}^2-{\lVert \bar{\mathbf{Y}} \rVert}^2)/\sigma_{\mathbf{Y}-g}^2$ follows a non-central chi-square distribution $\chi^2(\rho_g(\mathbf{X}_v, v\in g), \nu_g)$, with
\[
\rho_g(\mathbf{X}_v, v\in g) = \frac{\lVert \sum_{v\in g\cap A} a_v (\mathbf{X}_v - \bar {X_v}) \|^2}{\sigma^2_{\mathbf{Y}-g}}
\] 
where $\bar {X_v} = \frac{1}{n} \sum_{k=1}^n X_v^k {\mathbf{1}_n}$. Using the fact that 
$\rho_g(\mathbf{X}_v, v\in g)/n$ converges to 
\[
\rho_g :=  \frac{\sum_{v\in g\cap A} a_v^2 \sigma_v^2}{\sigma^2_{\mathbf{Y}-g}},
\] 
we will work with the approximation
%
%
%
%
\[
\frac{{\lVert P_g\mathbf{Y} \rVert}^2-{\lVert \bar{\mathbf{Y}} \rVert}^2}{\sigma_{\mathbf{Y}-g}^2}\sim \chi^2(n\rho_g,\nu_g).
\]
With a similar analysis, and letting for $v\in A$, 
$\sigma_{\mathbf{Y}-v}^2 = \sum_{v' \in A\setminus v}a_{v'}^2 \sigma_{v'}^2+\sigma^2$, we will assume that
\[
\frac{{\lVert P_v\mathbf{Y} \rVert}^2-{\lVert \bar{\mathbf{Y}} \rVert}^2}{\sigma_{\mathbf{Y}-v}^2}\sim \chi^2(n\rho_v,1)
\]
with 
\[
\rho_v :=  \frac{a_v^2 \sigma_v^2}{\sigma^2_{\mathbf{Y}-v}}.
\] 

We now have the simple lemma
\begin{proposition}
\label{prop:3}
If $Z_1\sim \chi^2(\rho_1, \nu_1)$ and $Z_2\sim \chi^2(\rho_2, \nu_2)$, 
 then for all $\theta_1, \theta_2$ such that $\theta_i \leq\rho_i+\nu_i$, i=1,2
\begin{equation}
\label{eq:power}
\mathbb{P}\left(Z_1>\theta_1; Z_2>\theta_2\right)\geq 
1-\sum_{i=1}^{2}\exp{\left(-\frac{{(\nu_i+\rho_i-\theta_i)}^2}{4(\nu_i+2\rho_i)}\right)}.
\end{equation}
\end{proposition}
\begin{proof}
This is based on the inequality \cite{LargeDev}, valid for $Z\sim \chi^2(\rho, \nu)$:
$$\mathbb{P}\left(Z \leq \rho+\nu -2 \sqrt{(\nu+2\rho)x}\right) \leq \exp(-x).$$

which implies
$$\mathbb{P}\left(Z \leq \theta\right) \leq \exp\left(-\frac{{(\nu+\rho-\theta)}^2}{4(\nu+2\rho)}\right)$$
as soon as $\theta\leq \rho+\nu$, and on the simple lower-bound
\[
\mathbb{P}\left(Z_1>\theta_1; Z_2>\theta_2\right)\geq 
1 - \sum_{i=1}^2 P(Z_i\leq\theta_i).
\]
\end{proof}

This proposition can be applied, in our case, to $(\rho_1, \nu_1) = (n\rho_g, |g|)$ and $(\rho_2, \nu_2) = (n\rho_v, 1)$. More concretely, we fix a target effect size $\eta$ (the ratio of the effect of $\mathbf{X}_v$ compared to the total variance of $\mathbf{Y}$), and a target cluster size, $k$, that represents the number of active loci that we expect to find in an active cell, and we take $\rho_v = \eta$ and $\rho_g = k\eta$ to optimize the upper-bound in \eqref{eq:power} subject to the FWER constraint \eqref{eq:cor.2} and $\tg \leq n\rho_g+|g|$ and $\tv \leq n\rho_v+1$ to find optimal constants $(\tg,\tv)$ for this target case. This is illustrated with numerical simulations in the next section.


\subsection{Simulation results (parametric case)} 

Figures \ref{fig:digraph.1}   compares the powers of the coarse-to-fine procedure  and of the Bonferroni-Holm procedure under the parametric model described in the section.
\begin{figure}
\centering
\includegraphics[width=.95\textwidth]{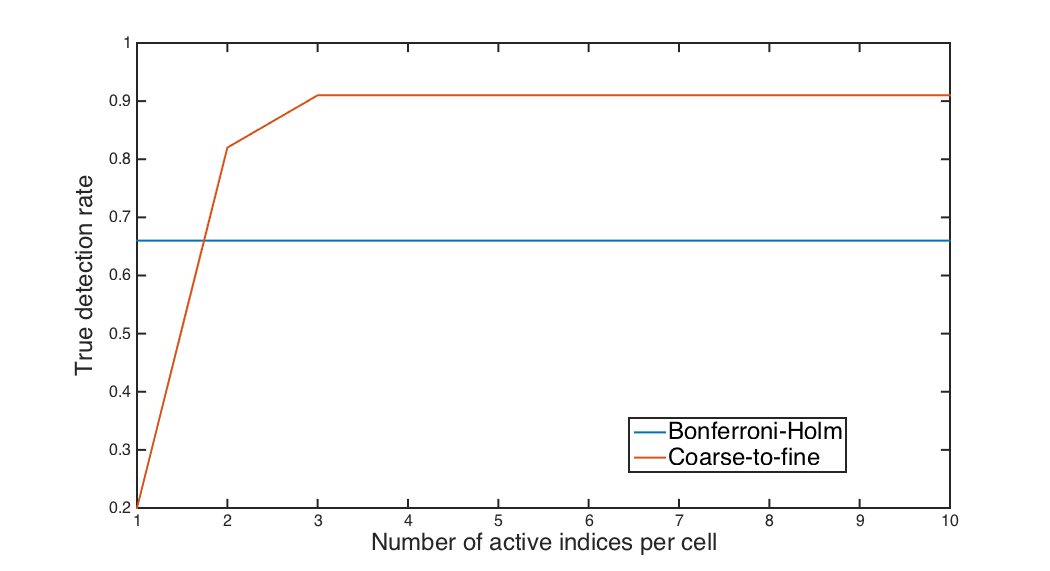}
\caption{\label{fig:digraph.1}Comparison of the power of different methods. We compare the detection rate of an individual active index for different number of active indices in the cell containing that index. The coarse-to-fine method is more powerful when the number of active indices is two or greater. This confirms the intuition that the more the clustering assumption is true, the more powerful is the coarse-to-fine method compared to Bonferoni-Holm approcach .}
\end{figure}

The parameters chosen in our simulations were taken with our motivating application (to GWAS) in mind. Thinking of $Y$ as a phenotype, and $V$ as a set of SNP's, we assimilate cells $g\in G$ to genes. 
We used  $|V|=10000$ and $|g|=10$ . The true number of active variables is 50 with a corresponding coefficient $a_v=1$ for each of them, and we generate the data according to the linear model described in this section with a variance noise that is equal to 10 . We assumed that we knew an upper bound $J$ for  the number of active sets (this assumption is relaxed in section \ref{sec:6}). To compute the optimal thresholds, some values for $\rho_g$ and $\rho_v$ have to be chosen (this should not be based on observed data, since this would invalidate our FWER and power estimates). In our experiments, we optimize the upper bound on the probability for an active variable to be detected in an active cell by choosing $\rho_g=2/J$ and $\rho_v=1/J$. This corresponds to an "almost non noisy case" where the effect size of the ``gene'' is two times the effect size of the ``SNP''.


\section{Non-parametric coarse-to-fine testing}
\label{sec:5}
\subsection{Notation}
Recall that $\mathbf U$ denotes the random variable representing all the data, taking values in $\U$. We will build our procedure from user-defined  {\em scores}, denoted $\rho_v$ (at the locus level) and $\rho_g$ (at the cell level), both defined on $\U$, i.e., functions of the observed data.

Moreover, we assume that there exists a group action of some group $\mathfrak S$ on $\U$, which will be denoted 
$$(\xi, \mathbf u) \mapsto \xi \bullet \mathbf u.$$
For example, if, like in the previous paragraph, one takes 
\[
\mathbf U=\left((Y^k,X^k), k=1, \ldots, n\right)
\]
  and $X^k = (X^k_v)_{v\in V}$, we will take $\mathfrak S$ to be the permutation group of $\{1, \ldots, n\}$ with 
$$\xi \bullet \mathbf U= \left((Y^{\xi_k},X^k), k=1, \ldots, n\right).$$
To simplify the discussion, we will assume that $\S$ is finite and denote by $\mu$ the uniform probability measure on $\S$, so that
\[
\int_\S f(\xi) d\mu(\xi) = \frac{1}{|\S|} \sum_{\xi\in \S} f(\xi).
\]
We note, however, that our discussion remains true if $\S$ is a compact group, $\mu$ the right-invariant Haar probability measure on $\S$ and $\rho_v(\xi\bullet \mathbf u), \rho_g(\xi\bullet \mathbf u)$ are continuous in $\xi$. 

 Our running assumption will be that,
\begin{enumerate}
\item For any $v\in V_{00}$, the joint distribution of $(\rho_{g(v)}((\xi'\xi)\bullet \bU), \rho_v((\xi'\xi)\bullet \bU))_{\xi'\in\S}$ is independent of $\xi\in \S$.
\item For any $v\in V_{0}$, the joint distribution of $(\rho_v((\xi'\xi)\bullet \bU))_{\xi'\in \S}$ is independent of $\xi\in \S$.
\end{enumerate}


%

We will also use the following well-known result. 
 \begin{lemma}
 \label{lem:basic}
 Let $X$ be a random variable and let $F^-_X$ denote the left limit of its cumulative distribution function, i.e.,  $F^-_X(t) = P(X<t)$. Then, for $t\in [0,1]$, one has
 $$P\left(F^-_X(X) \geq 1-t\right)\leq t$$
(with equality if $F$ is continuous).
 \end{lemma}

\subsection{Asymptotic resampling scores}
We  define the asymptotic scores at the cell and variable  level by
\begin{equation}
\label{eq:tg}
T_g(\bu)=\mu\left(\xi: \rho_g(\bu)\leq\rho_g(\xi \bullet \bu)\right)
\end{equation}
and
\begin{equation}
\label{eq:tv}
T_v(\bu)= \mu\left(\xi: \rho_v(\bu)\leq\rho_v(\xi \bullet \bu)\right).
\end{equation}
$T_g(\bU)$ and $T_v(\bU)$ are the typical statistics used in randomized tests, estimating the proportion of scores that are higher than the observed one after randomization. 
For the coarse-to-fine procedure, we will need one more ``conditional'' statistic. 

For a given constant $\tg$, we define 
\begin{equation}
\label{eq:ng}
N_g^{\tg}(\bu) =  \mu\left(\xi: T_g(\xi \bullet \bu)\leq\tg\right).
\end{equation}
We then let
\begin{equation}
\label{eq:tv.2}
 T_{v}^{\tg}(\bu)=\frac{1}{N_{g(v)}^{\tg}(\bu)} \mu\left(\xi: \rho_v(\bu)\leq\rho_v(\xi \bullet \bu); T_{g(v)}(\xi \bullet \bu)\leq\tg\right).
 \end{equation}
We call our scores asymptotic in this section because exact expectations over $\mu$ cannot be computed in general, and can only be obtained as limits of Monte-Carlo samples. The practical finite-sample case will be handled in the next section.

We this notation, we let
\[
\hat A = \{v: T_{g(v)}(\bU)\leq\tg \text{ and } T_{v}^{\tg}(\bU)\leq\tv \text{ and }T_v(\bU)\leq\tv'\}
\]
which depends on the choice of three constants, $\tv, \tg$ and $\tv'$. We then have:

\begin{theorem}
\label{th:ideal}
For all $v\in V_0$:
\begin{equation}
\label{eq:ideal.0}
\mathbb{P}\left(v \in \hat{A}\right) \leq \tv'
\end{equation}
and for all $v\in V_{00}$,
\begin{equation}
\label{eq:ideal.00}
\mathbb{P} \left(v \in \hat{A}\right) \leq \tg\tv
\end{equation}
\end{theorem}

This result tells us how to control the FWER for a two-level permutation test based on any scores in the (generally intractable) case in which we can exactly compute the test statistics, when we declare an index $v$ active if and only if $T_g(\bU)\leq\tg$ and $T_{v}^{\tg}(\bU)\leq\tv$ and  $T_{v}(\bU)\leq\tv'$. 
\begin{proof}
For \eqref{eq:ideal.0}, we use a standard argument justifying randomization tests, that we provide here for completeness.
If $v\in V_0$, we have
\begin{align*}
 \mathbb{P}\left(v \in \hat{A}\right) &= \mathbb{P}\left(T_g(\bU)\leq\tg;T_{v}^{\tg}(\bU)\leq\tv; T_{v}(\bU)\leq\tv'\right)
   \\                                                  &\leq \mathbb{P}\left( T_{v}(\bU)\leq\tv' \right).
\end{align*}\\
From the invariance assumption, we have
\begin{eqnarray*}
\mathbb{P}\left( T_{v}(\bU)\leq\tv' \right) &=& \mathbb{P}\left( T_{v}(\xi\bullet \bU)\leq\tv' \right) \text{ for all } \xi\in \S\\
&=& \int_\S \mathbb{P}\left( T_{v}(\xi\bullet \bU)\leq\tv' \right) d\mu(\xi)\\
&=& \mathbb E \left(\mu\left(\xi: T_{v}(\xi\bullet \bU)\leq\tv' \right)\right)
\end{eqnarray*}
It now remains to remark that $T_v(\xi\bullet \bU) = 1- F_\zeta^-(\zeta(\xi))$ where $\zeta$ is the random variable on $\S$ defined by $\zeta(\xi') = \rho_v(\xi'\bullet \bU)$,  so that, by Lemma \ref{lem:basic},
\[
\mu\left( \xi: T_{v}(\xi\bullet \bU)\leq\tv' \right) = \mu\left(\xi: F_\zeta^-(\zeta(\xi)) \geq 1-\tv'\right) \leq \tv',
\]
which proves \eqref{eq:ideal.0}.\\


Let us now prove \eqref{eq:ideal.00}, assuming $v\in V_{00}$ and letting $g=g(v)$. We write
\begin{align*}
\mathbb{P}\left(v \in \hat{A}\right)  & \leq \mathbb{P}\left(T_{g(v)}(\bU)\leq\tg;T_{v}^{\tg}(\bU)\leq\tv \right).                                                           
 \end{align*}
 and find an upper bound for the right-hand side of the inequality.
Using the invariance assumption, we have, for all $\xi \in G$,
\begin{align*}
\mathbb{P}\left(T_g(\bU)\leq\tg;T_{v}^{\tg}(\bU)\leq\tv \right) &=\mathbb{P}\left(T_g(\xi \bullet \bU)\leq\tg;T_{v}^{\tg}(\xi \bullet \bU)\leq\tv \right)
\\                                                                                  &=\int_\S \mathbb{P}\left(T_g(\xi^{'} \bullet \bU)\leq\tg;T_{v}^{\tg}(\xi^{'} \bullet \bU)\leq\tv \right) d\mu(\xi')
\\                                                                                  &=\mathbb{E}\left(\mu\left(\xi': T_g(\xi^{'} \bullet \bU)\leq\tg; T_{v}^{\tg}(\xi^{'} \bullet \bU)\leq\tv\right)\right).
\end{align*}
Notice that, since $\mu$ is right-invariant, we have $N_g^{\tg}(\xi'\bullet \bU) = N_g^{\tg}(\bU)$ and 
\begin{align*}
T_{v}^{\tg}(\xi^{'} \bullet \bU) &=\frac{1}{N_g^{\tg}(\xi'\bullet \bU)} \mu\left(\xi: \rho_v(\xi^{'} \bullet \bU)\leq\rho_v(\xi \bullet \xi^{'} \bullet \bU);T_g(\xi \bullet \xi' \bullet  \bU)\leq\tg\right)
\\                                                 &=\frac{1}{N_g^{\tg}( \bU)} \mu\left(\xi:\rho_v(\xi^{'} \bullet \bU)\leq\rho_v(\xi \bullet \bU); T_g(\xi \bullet \bU)\leq\tg\right).
\end{align*}

Let $\tilde \mu$ denote the  probability $\mu$ conditional to the event $T_g(\xi \bullet \bu)\leq\tg$ ($\bu$ being fixed). Then
$$\frac{1}{N_g^{\tg}( \bu)}\mu\left (T_g(\xi' \bullet \bu)\leq\tg;T_{v}^{\tg}(\xi' \bullet \bu)\leq\tv\right)={\tilde \mu} \left(\xi': {p}(\xi')\leq\tv\right),$$
where
$${p}(\xi')=\tilde \mu\left(\xi: \rho_v(\xi \bullet \bu)\geq \rho_v(\xi' \bullet \bu)\right).$$
Hence,  Lemma \ref{lem:basic} implies that, for each $\bu$:
$$\frac{1}{N_g^{\tg}( \bu)}\mu\left(\xi': T_g(\xi' \bullet \bu)\leq\tg;T_{v}^{\tg}(\xi^{'} \bullet \bu)\leq\tv\right)\leq \tv.$$
Hence,
\[
\mathbb{P}\left(T_g(\bU)\leq\tg;T_{v}^{\tg}(\bU)\leq\tv \right) \leq \mathbb{E}\left(N_g^{\tg}( \bU)\tv\right)
 = \tv \mathbb{E}\left(N_g^{\tg}( \bU)  \right).
\]

Applying Lemma \ref{lem:basic} to the random variable $\xi \mapsto \rho_g(\xi)$ for the probability distribution $\mu$, we immediately get $N_g^{\tg}( \bU) \leq \tg$
so that
$$\mathbb{P}\left(T_g(\bU)\leq \tg;T_{v}^{\tg}(\bU)\leq \tv \right) \leq \tg \tv.$$

\end{proof}

As an immediate corollary, we have:
\begin{corollary}
\label{cor:fwer.ideal}
$$\FWER(\hat A) \leq  |V| \tg \tv+J \nu_g \tv'.$$
\end{corollary}

\begin{remark}[Continuous approximation]
\label{rem:1}
Even though $\S$ is finite, it is a huge set in typical applications, and while Lemma \eqref{lem:basic} only provides inequalities for discrete distributions, we can safely ignore the discontinuity in practice and work as if the distributions to which we applied this lemma were continuous. Doing so, it is easy to convince oneself, by inspecting the previous proof that our estimates become (for $v\in V_{00}$, $g=g(v)$)
$$\mathbb{P}\left(T_g(\bU)\leq \tg;T_{v}^{\tg}(\bU)\leq \tv \right) = \tg \tv.$$
This implies that
$$\mathbb{P}\left(T_{v}^{\tg}(\bU)\leq \tv| T_g(\bU) \leq \tg\right) = \tv,$$
because we have also:
$$\mathbb{P}\left(T_g(\bU)\leq \tg \right) = \tg.$$
This tells us that, conditional to  $T_g(\bU) \leq \tg $, $T_{v}^{\tg}(\bU)$ is uniform distributed on $[0,1]$.
\end{remark}

As mentioned above, this result does not have practical interest since it requires applying all possible permutations to the data. In practice, a random subset of permutations is picked instead, and we will develop the related theory in the next section (using these inequalities as intermediary results in our proofs).

\subsection{Finite resampling scores}
We now replace $T_g$, $T_v^{\tg}$ and $T_v$ with Monte-Carlo estimates and describe how the upper bounds in Theorem \ref{th:ideal} need to be modified. 

For a positive integer $K$, we let $\mu^K$ denote the $K$-fold product measure on $\S$, whose realizations are $K$ independent group elements  $\bxi  = (\xi_1,...,\xi_K)\in \S^K$.  We will use the notation $\boldsymbol P$ and $\boldsymbol E$ to denote probability or expectation for the joint distribution of $U$ and $\xi$ (i.e., $\boldsymbol P = \mP\otimes \mu^K$). We will also denote by $\hat\mu_{\bxi}$ the empirical measure
\[
\hat\mu_{\bxi}  = \frac1K \sum_{k=1}^K \delta_{\xi_k}.
\]
With this notation, we let
$$\hat{T}_g(\bu, \bxi)=\hat \mu_{\bxi}\left(\xi': \rho_g(\bu)\leq\rho_g(\xi' \bullet \bu)\right) = \frac1K \sum_{k=1}^K {\mathbf 1}_{\rho_g(\bu)\leq\rho_g(\xi_k \bullet \bu)},
$$
$$\hat{T}_v(\bu)= \hat\mu_{\bxi}\left(\xi': \rho_v(\bu)\leq\rho_v(\xi' \bullet \bu)\right)$$
and
$$\hat{T}_{v}^{\tg,\eg}(\bu, \bxi)= \frac1{\hat N_g^{\tg,\eg}(\bu, \bxi)}\hat\mu_{\bxi}\left(\xi': \rho_v(\bu)\leq\rho_v(\xi' \bullet \bu);\hat{T}_g(\xi' \bullet \bu, \bxi)\leq(\tg+2\eg)\right)$$
where
$$\hat N_g^{\tg,\eg}(\bu, \bxi)= \hat\mu_{\bxi} \left(\xi': \hat{T}_g(\xi' \bullet \bu)\leq(\tg+\eg)\right).$$

We can now define
\[
\hat A = \{v: \hat T_{g(v)}(\bU)\leq\tg \text{ and } \hat T_{v}^{\tg,\eg}(\bU)\leq\tv \text{ and } \hat T_v(\bU)\leq\tv'\}
\]
and state:
%
%
\begin{theorem}
\label{th:sample}
Making the continuous approximation described in Remark \ref{rem:1}, the following holds.
For $v\in V_0$,
\begin{equation}
\label{eq:sample.0}
\boldsymbol{P}\left(v \in \hat{A} \right) \leq  \tv'.
\end{equation}
and, for $v\in V_{00}$ and $g=g(v)$,
\begin{equation}
\label{eq:sample.00}
\boldsymbol{P}\left(v \in \hat{A}\right)  \leq (K+1)\exp{\left(-2(K-1)\left(\eg-\frac{1}{K-1}\right)^2\right)}+\frac{K(\tg+\eg)\tv+1}{K+1}
\end{equation}

\end{theorem}


\begin{corollary}
\label{cor:sample} The FWER for the randomized test is controlled by
\begin{multline*}
\FWER \leq |V|\left( \exp{(-2K\eg^2)}+K\exp{\left(-2(K-1)\left(\eg-\frac{1}{K-1}\right)^2\right)}\right.\\
\left.+\frac{K(\tg+\eg)\tv+1}{K+1} \right)
+J \nu_g \tv'.
\end{multline*}
\end{corollary}
 
\begin{proof}[Proof of Theorem \ref{th:sample}]

We start with \eqref{eq:sample.0} which is simpler and standard. Let $v\in V_0$. Conditionally to $\bU$, $K\hat T_v(\bU, \bxi)$ follows a binomial distribution $\mathrm{Bin}(K, T_v(\bU))$ (with $T_v$ defined by equation \eqref{eq:tv}), so that
\[
\boldsymbol P(\hat T_v \leq \tv') = \sum_{j=0}^{\floor{K\tv'}} \binom{K}{j} \mathbb E\left(T_v(\bU)^j(1-T_v(\bU))^{K-j}\right)
\]
where $\floor{\cdot }$ denotes the integral part.
The continuous approximation implies that $T_v(\bU)$ is uniformly distributed, so that
\[
\mathbb E\left(T_v(\bU)^j(1-T_v(\bU))^{K-j}\right) = \int_0^1 t^j(1-t)^{K-j}dt = \frac{j!(K-j)!}{(K-1)!}
\]
yielding 
\[
\boldsymbol P(\hat T_v \leq \tv') = \frac{\floor{K\tv'}}{K} \leq \tv'.
\]
\\

We now consider \eqref{eq:sample.00} and take $v\in V_{00}$, $g=g(v)$.
We need to prove that:
\begin{multline*}
\boldsymbol{P}\left(\hat{T}_g(\bU, \bxi) \leq \tg;\hat{T}_{v}^{\tg,\eg}(\bU,\bxi) \leq\theta_v\right)  \\
\leq \exp{(-2K\eg^2)}+K\exp{\left(-2(K-1)\left(\eg-\frac{1}{K-1})\right)^2\right)}+\frac{K(\tg+\eg)\tv+1}{K+1}.
\end{multline*}

Given $\bU$, 
$\hat{T}_g(\bU,\bxi)$ is the empirical mean of $K$ i.i.d Bernoulli random  variables ($Y_k= \mathbf{1}_{\rho_g(\bU)\leq\rho_g(\xi_k \bullet \bU)}$) with success probability $T_g(\bU)$.
Hoeffding's inequality \cite{Hoeffding} implies that
$$\boldsymbol{P}\left(\hat{T}_g(\bU,\bxi) \leq T_g(\bU)-\eg|\bU\right) \leq \exp{(-2K\eg^2)}.$$
So:
\begin{align}
\nonumber
\boldsymbol{P}\left(\hat{T}_g(\bU,\bxi) \leq \tg;\hat{T}_{v}^{\tg,\eg}(\bU) \leq \tv\right)&=\mathbb{E}\left(\boldsymbol{P}\left(\hat{T}_g(\bU,\bxi) \leq \tg;\hat{T}_{v}^{\tg,\eg}(\bU,\bxi) \leq \tv|\bU\right)\right)
\\                 
\nonumber                                                                                                                                                   & \leq \mathbb{E}\left(\boldsymbol{P}\left({T}_g(\bU) \leq \tg;\hat{T}_{v}^{\tg,\eg}(\bU,\bxi) \leq \tv|\bU\right)\right)\\
\nonumber
&\qquad +\exp{(-2K\eg^2)}
 \\                
 \label{eq:step.1}                                                                                                                                                     & =  \boldsymbol{P}\left({T}_g(\bU) \leq \tg+\eg;\hat{T}_{v}^{\tg,\epsilon_g}(\bU,\bxi) \leq \tv\right)\\
 \nonumber & \qquad +\exp{(-2K\eg^2)}
\end{align}

We now fix $i_0$ in $1,...,K$ and consider
$$\boldsymbol{P}\left(\hat{T}_g(\xi_{i_0} \bullet \bU, \bxi)\geq T_g(\xi_{i_0} \bullet \bU)+\eg|\bU,\xi_{i_0}\right).$$
Since this probability does not depend on which $i_0$ is chosen, we will estimate it, without loss of generality, for $i_0=1$, which will simplify the notation.
For this, notice that
$$\hat{T}_g(\xi_{1} \bullet \bU, \bxi)=\frac{1}{K}+\frac{1}{K}\sum_{i=2}^{K}\mathbf{1}_{\rho_g(\xi_i \bullet \bU) \geq \rho_g(\xi_{1} \bullet \bU)},$$
and also that
\[
\boldsymbol{E}\left(\mathbf{1}_{\rho_g(\xi_i \bullet \bU) \geq \rho_g(\xi_{1} \bullet \bU)}\right | \bU, \xi_{1}) =\mu\left(\xi: \rho_g(\xi\bullet \bU) \geq \rho_g(\xi_1 \bullet \bU)\right)
 =T_g(\xi_{1} \bullet \bU).
 \]
Now we use Hoeffding's inequality to obtain
\begin{align*}
& \boldsymbol{P}\left(\hat{T}_g(\xi_1 \bullet \bU,\bxi)\geq T_g(\xi_1 \bullet \bU)+\eg|\bU,\xi_1\right)\\ 
&\qquad = \boldsymbol P\left(\frac{1}{K-1}\sum_{i=2}^{K}\mathbf{1}_{\rho_g(\xi_i \bullet \bU) \geq \rho_g(\xi_{1} \bullet \bU)} \leq \frac{K\eg-1}{K-1} \right)
\\
&\qquad \leq \exp{\left(-2(K-1)\left(\eg-\frac{1}{K-1}\right)^2\right)}.
\end{align*}
The same upper-bound applies to $\boldsymbol{P}\left(\hat{T}_g(\xi_1 \bullet \bU,\bxi)\geq T_g(\xi_1 \bullet \bU)+\eg|\bU\right)$ by taking the expectation with respect to $\xi_1$, and one deduces from this that
\begin{align*}
& \boldsymbol{P}\left(\exists i: \hat{T}_g(\xi_i \bullet \bU,\bxi)\geq T_g(\xi_i\bullet \bU)+\eg\right)\\
&\qquad = \mathbb E\left(\boldsymbol{P}\left(\exists i: \hat{T}_g(\xi_i \bullet \bU,\bxi)\geq T_g(\xi_i\bullet \bU)+\eg|\bU\right)\right)\\
&\qquad \leq K\exp{\left(-2(K-1)\left(\eg-\frac{1}{K-1}\right)^2\right)}.
\end{align*}
Introduce
$$\hat{T}_{v}^{*,\tg,\eg}(\bu, \bxi)= \frac1{\tg+\eg}\hat\mu_{\bxi}\left(\xi': \rho_v(\bu)\leq\rho_v(\xi' \bullet \bu);  T_g(\xi' \bullet \bu)\leq(\tg+\eg)\right).$$
On the event 
\[
C = \left(\hat{T}_g(\xi_i \bullet \bU,\bxi)\leq T_g(\xi_i\bullet \bU)+\eg, i=1, \ldots, K\right)
\]
one has 
\[
\hat T_v^{*,\tg,\eg}(\bU,\bxi) \leq \frac1{\tg+\eg}\hat\mu_{\bxi}\left(\xi': \rho_v(\bU)\leq\rho_v(\xi' \bullet \bU);  \hat T_g(\xi' \bullet \bU,\bxi)\leq(\tg+2\eg)\right).
\]
Moreover, we have
$\hat N_g^{\tg,\eg}(\bU, \bxi) \leq \tg + \eg$, by applying Lemma \ref{lem:basic} to the random variable $\xi \mapsto \rho_g(\xi\bullet \bU)$ under the distribution $\hat \mu^K_{\bxi}$. This implies that, on $C$, we have
\[
\hat T_v^{*,\tg,\eg}(\bU,\bxi) \leq \hat{T}_{v}^{\tg,\eg}(\bU,\bxi)
\]
which implies
\begin{multline}
\label{eq:step.2}
\boldsymbol{P}\left({T}_g(\bU) \leq \tg+\eg;\hat{T}_{v}^{\tg,\eg}(\bU,\bxi) \leq \tv\right) \\
\leq K\exp{\left(-2(K-1)\left(\eg-\frac{1}{K-1}\right)^2\right)} + \boldsymbol{P}\left({T}_g(\bU) \leq \theta_g;\hat{T}_{v}^{*,\tg,\eg}(\bU,\bxi) \leq \tv\right).
\end{multline}
%
%

We now provide an estimate for the last term in \eqref{eq:step.2}. We have
\begin{align}
\nonumber
&\boldsymbol{P}\left({T}_g(\bU) \leq \tg+\eg;\hat{T}_{v}^{*,\tg,\eg}(\bU,\bxi) \leq \tv\right) \\
\nonumber
&\qquad =\boldsymbol{P}\left(\hat{T}_{v}^{*,\tg,\eg}(\bU,\bxi) \leq \tv|{T}_g(\bU) \leq \tg+\eg\right) 
\mP\left({T}_g(\bU) \leq \tg +\eg\right)
\\                 
&\qquad  \leq (\tg+\eg) \times \boldsymbol{P}\left(\hat{T}_{v}^{*,\tg,\eg}(\bU,\bxi) \leq \tv|{T}_g(\bU) \leq \tg+\eg\right).
\label{eq:step2.2}                                                                                                                                                        
\end{align}
Now write
\begin{align*}
&\boldsymbol{P}\left(\hat{T}_{v}^{*,\tg,\eg}(\bU,\bxi) \leq \tv|{T}_g(\bU) \leq \tg+\eg\right) \\
&\qquad =\boldsymbol P \left(\hat\mu_{\bxi}^K(\xi': \rho_v(\bU)\leq\rho_v(\xi' \bullet \bU);  \right.\\
&\qquad\qquad\qquad\qquad \left.{T}_g(\xi' \bullet \bU)\leq(\tg+\eg)) \leq \tv(\tg+\eg) | {T}_g(\bU) \leq \tg+\eg\right)\\
&\qquad =\mathbb E\left(\boldsymbol P \left(\hat\mu_{\bxi}^K(\xi': \rho_v(\bU)\leq\rho_v(\xi' \bullet \bU);  \right.\right.\\
& \qquad\qquad\qquad\qquad \left.\left.{T}_g(\xi' \bullet \bU)\leq(\tg+\eg)) \leq \tv(\tg+\eg) \big| \bU\right) \Big |  {T}_g(\bU) \leq \tg+\eg\right)
\end{align*}
Given $\bU$, the probability (for $\mu$) of the event 
\[
\left(\rho_v(\bU)\leq\rho_v(\xi' \bullet \bU);  {T}_g(\xi' \bullet \bU)\leq(\tg+\eg)\right)
\]
 is $p(\bU) = N_g^{\tg+\eg}(\bU)T_v^{\tg+\eg}(\bU)$ and
\[
K\hat\mu_{\bxi}(\xi': \rho_v(\bU)\leq\rho_v(\xi' \bullet \bU);  {T}_g(\xi' \bullet \bU)\leq(\tg+\eg))
\]
follows a binomial distribution $\mathrm{Bin}(K, p(\bU))$. We make the continuous approximation (Remark \ref{rem:1}) to write $N_g^{\tg+\eg}(\bU) = (\tg+\eg)$, and to use the fact that $T_v^{\tg+\eg}(\bU)$ is uniformly distributed on $[0,1]$ conditionally to $T_g(\bU) \leq \tg+\eg$, so that 
%
$p(\bU)$ is uniformly distributed on $[0,\tg+\eg]$ given  ${T}_g(\bU) \leq \tg$. This implies
\begin{align*}
&\boldsymbol{P}\left(\hat{T}_{v}^{*,\tg,\eg}(\bU,\bxi) \leq \tv|{T}_g(\bU) \leq \tg+\eg\right)  
\\
&\qquad =
\frac{1}{\tg+\eg}\int_{0}^{\tg+\eg}\left(\sum_{i=0}^{\floor{K(\tg+\eg)\tv}}\binom{K}{ i} p^{i}{(1-p)}^{K-i}\right)dp\\
&\leq
\frac{1}{(K+1)(\tg+\eg)}\sum_{i=0}^{\floor{K(\tg+\eg)\tv}}\left(\int_{0}^{1}\binom{K+1}{ i} p^{i}{(1-p)}^{K-i}dp\right)\\
&
 = \frac{\floor{K(\tg+\eg)\tv}+1}{(K+1)(\tg+\eg)}
\end{align*}
since each of the integrals is equal to 1 (they are densities of Beta distributions). 
From \eqref{eq:step2.2}, we find
\begin{equation}
\label{eq:step.3}
\boldsymbol{P}\left(\hat{T}_{v}^{*,\tg,\eg}(\bU,\bxi) \leq \tv; {T}_g(\bU) \leq \tg+\eg\right) \leq \frac{K(\tg+\eg)\tv+1}{K+1},
\end{equation}
which, combined with equations \eqref{eq:step.1} and \eqref{eq:step.2} completes the proof of the theorem.
\end{proof}

\subsection{Simulations}

For the simulations, we generated data according to the parametric model described  in section \ref{sec:3}, and compared the detection rate using the non-parametric approach to the one obtained with thresholds optimized as described in section \ref{sec:3}. We fixed the ratio between the  non-parametric thresholds $\tg$ and $\tv$ to coincide with the ratio between the type I errors at the cell and index levels in the parametric case, i.e.,
\[
\frac{\tv}{\tg} = \frac{\mathbb P(T_v^{\mathit{par}} > \tv^{\mathit{par}})}{\mathbb P(T_g^{\mathit{par}} > \tg^{\mathit{par}})}
\]
for $v\in V_{00}$, where $T_v^{\mathit{par}}, T_v^{\mathit{par}}, \tv^{\mathit{par}}, \tg^{\mathit{par}} $ are the scores and thresholds designed for the parametric case. Fixing $\tv/\tg$ and the FWER uniquely determines the thresholds. 

Figure \ref{fig:digraph.0} provides a comparison between the coarse-to-fine approaches (parametric and non-parametric) and Bonferroni-Holm. In figures \ref{fig:digraph.4} and \ref{fig:digraph.5}, the non-parametric coarse-to-fine method is compared with Bonferroni-Holm for two different index-level effect sizes. Using the notation of section \ref{sec:3}, the cell-level effect size is defined by
\[
\text{effect size}(g) = {\sum_{v\in g} a_v^2 \sigma_v^2}/{\sigma_{\mathbf Y}^2}
\]
while $\text{effect size}(v) = a_v^2\sigma_v^2/\sigma_{\mathbf Y}^2$ at the index level. The range of index-level effect sizes we consider in these simulation is similar to the one observed in typical genome-wide association studies.


\begin{figure}
\centering
\includegraphics[width=.95\textwidth]{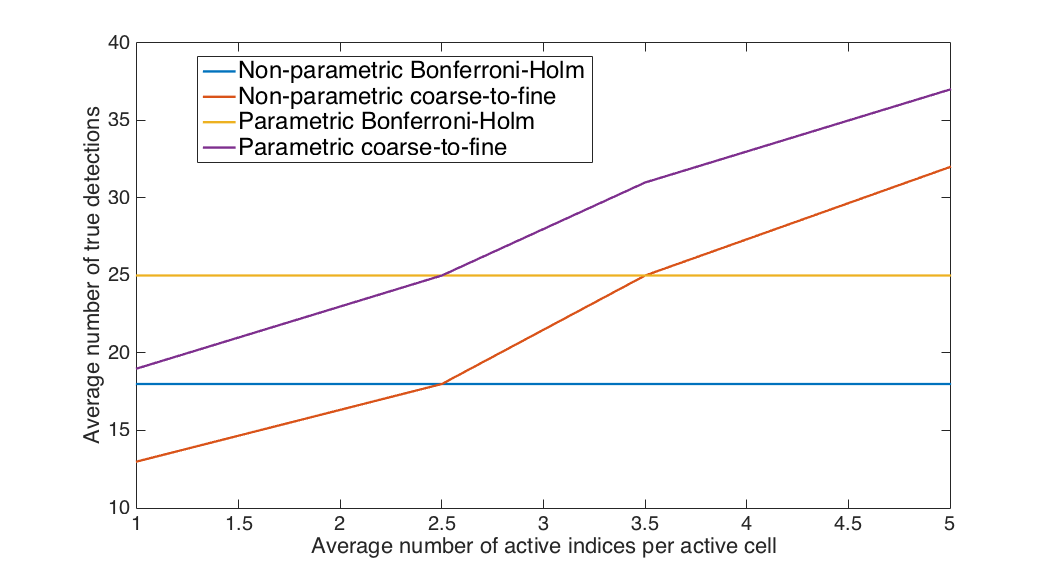}
\caption{Comparison of the  coarse-to-fine and Bonferroni-Holm methods in the parametric and non-parametric methods when data is generated from the parametric model.}
\label{fig:digraph.0}
\end{figure}

\begin{figure}
\centering
\includegraphics[width=.95\textwidth]{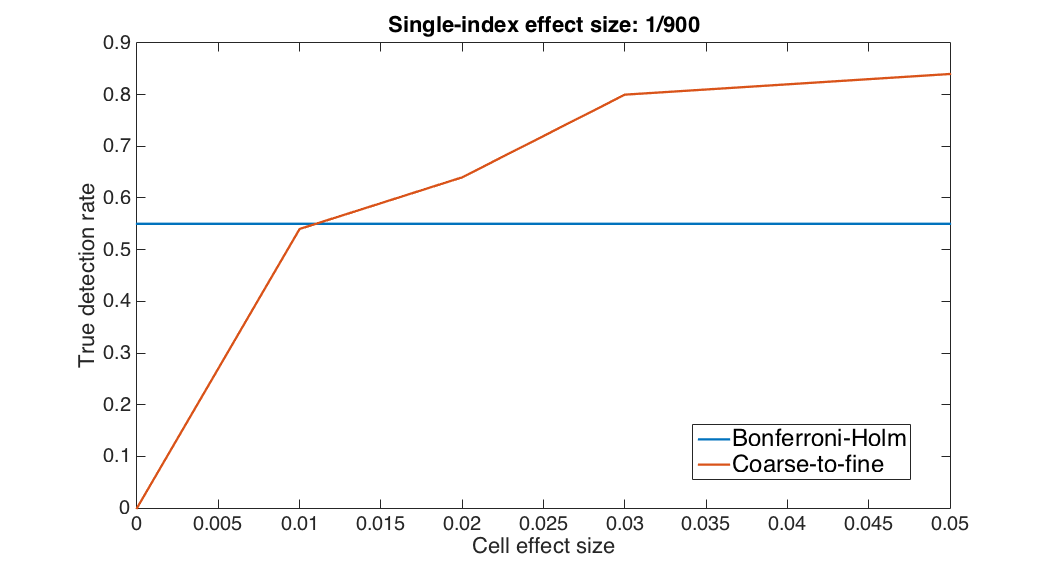}
\caption{Comparison of the non-parametric Bonferroni-Holm and coarse-to-fine methods for an index-level effect size equal to 1/900. One can notice that when the clustering assumption is true, the coarse-to-fine method outperforms the  Bonferroni-Holm approach}
\label{fig:digraph.4}
\end{figure}

\begin{figure}
\centering
\includegraphics[width=.95\textwidth]{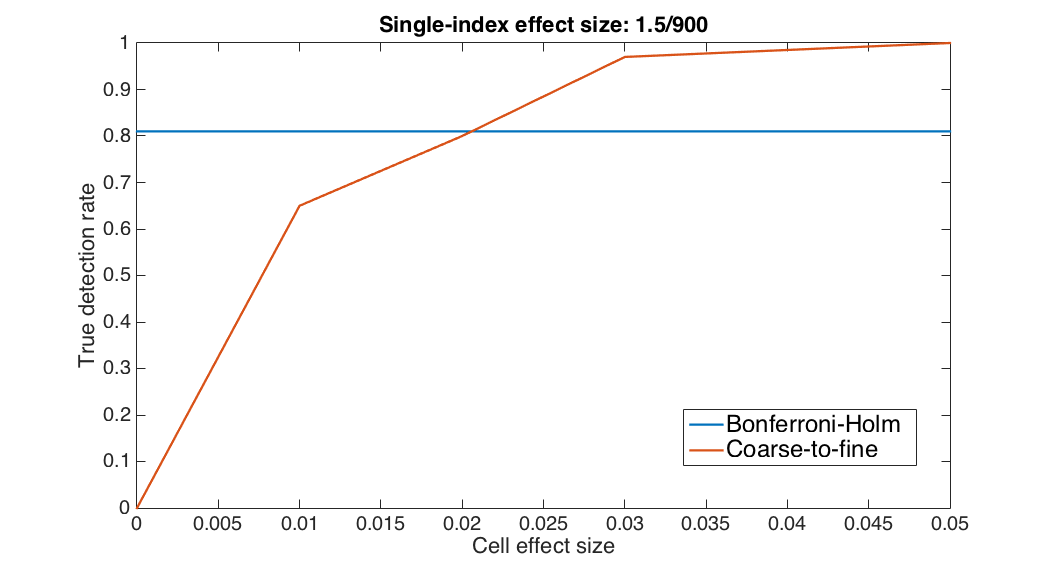}
\caption{Comparison of the non-parametric Bonferroni-Holm and coarse-to-fine methods for an index-level effect size equal to 1.5/900. A larger effect size (compared to figure \ref{fig:digraph.4} improves the performance of  the Bonferroni-Holm method.  The coarse-to-fine method is performs better at levels that correspond to two or more active indices per cell.}
\label{fig:digraph.5}
\end{figure}

\section{Estimating the number of active cells}
\label{sec:6}
\subsection{Notations, assumptions and starting point}

We now focus on the issue of estimating the number of active cells, $J$, from observed data, since this number intervenes in  our FWER estimates. We use a method inspired from \cite{Storey2001,storey2002direct} for the estimation of false discovery rates, adapted to our context. 
Our estimation will be made based on cell statistics $(T_g, g\in G)$ under the following setting. 
\begin{itemize}
\item[A1.] If $g\cap A = \emptyset$ ($g$ is inactive), then $T_g$ is uniformly distributed.
\item[A2.] If $g_1, g_2$ are inactive, then $T_{g_1}$, $T_{g_2}$ are independent.
\end{itemize}
These assumptions will be justified in section \ref{sec:justify}.
We will also assume that $T_g$ takes large values when $g$ is active, so that, for a suitable non conservative threshold $t_0$, we have $P(T_g \geq t_0) \simeq 1$. To simplify the argument, we will actually make the approximation that:
\begin{itemize}
\item[A3.] There exists $t_0\in (0,1)$ such that $P(T_g \geq t_0) = 1$  if $g\cap A \neq \emptyset$. 
\end{itemize}


For $t\in [0,1]$, we define $\hat D_t = \{g: T_g \leq t\}$. Let $D$ be the set of active cells, and $G_0 = G\setminus D$. Note that $J = |D|$.
Then, for $t\geq t_0$,
\begin{eqnarray*}
  \mathbb{E}\left(|\hat{D}_t|\right)   & =& \mathbb{E} \left(   \sum_{g \in G} \mathbf{1}_{g \in \hat{D}_t}       \right)\\ 
  &=&
    \mathbb{E} \left(   \sum_{g \in G_0}   \mathbf{1}_{g \in \hat{D}_t} \right)  +    \mathbb{E} \left(   \sum_{g \in D}   \mathbf{1}_{g \in \hat{D}_t}\right)\\
    &=&
    (|G| - J)t+J.
    \end{eqnarray*}
 We therefore have, for $t\geq t_0$,
 $$  |\hat{D}_t|= (|G|-J )t+J+Z_t,$$
where $Z_t$ is a centered random variable.  The following proposition states that the process $\frac{Z_t}{\sqrt{|G|-J}}$ for $t\geq t_0$ has the covariance structure of a Brownian bridge. \\

\begin{proposition}
\label{prob:cov}
Under assumptions A1 to A3, we have, for $t_1, t_2 \geq t_0$,
\[
\cov(Z_{t_1}, Z_{t_2}) = (|G|-J)(\min(t_1,t_2) - t_1t_2)
\]
\end{proposition}
\begin{proof}
 Since
$Z_{t}=\sum_{g \in G_0}   \big(\mathbf{1}_{g \in \hat{D}_t} -t \big)$,
one has
\begin{equation}
\label{eq:cov.1}
\cov(Z_{t_1},Z_{t_2})=\sum_{g_1, g_2 \in G_0}\cov(\mathbf{1}_{g_1 \in \hat{D}_{t_1}},\mathbf{1}_{g_2 \in \hat{D}_{t_2}}).
\end{equation}
 
If $g_1 \neq g_2$, then $\cov(\mathbf{1}_{g_1 \in \hat{D}_{t_1}},\mathbf{1}_{g_2 \in \hat{D}_{t_2}})=0$,
and for $g_1=g_2(=g)$:
\begin{eqnarray*}
 \cov(\mathbf{1}_{g \in \hat{D}_{t_1}},\mathbf{1}_{g\in \hat{D}_{t_2}})
 &=& \mathbb{E}\left(\mathbf{1}_{g \in \hat{D}_{\min(t_1,t_2)}}\right)-\mathbb{E}\left(\mathbf{1}_{g \in \hat{D}_{t_1}}\right)\mathbb{E}\left(\mathbf{1}_{g \in \hat{D}_{t_2}}\right)\\
&=& \min(t_1,t_2)-t_1 t_2.
\end{eqnarray*} 
 Finally, from \eqref{eq:cov.1}, we get
  \[
  \cov(Z_{t_1},Z_{t_2})=(|G|-J)\big(\min(t_1,t_2)-t_1 t_2 \big)
  \]
  which concludes the proof.
 \end{proof}

We now make a Gaussian approximation for large $G$ of the vector $\frac{\zeta}{\sqrt{(|G|-J)}}$ with $\zeta = (Z_{t_1}, \ldots, Z_{t_k})$, with $t_0\leq t_1 < \cdots<t_k\leq 1$.
  \begin{proposition}
  \label{prop:clt}
  Using the previous notation, 
 $\frac{\zeta}{\sqrt{(|G|-J)}}\overset{\mathcal{L}}{\longrightarrow} \mathcal{N}(0,\Gamma)$ when $|G|$ diverges to infinity, where $\Gamma$ is the covariance matrix with entries:
 $$\Gamma(i,j)=min(t_i,t_j)-t_i t_j $$
  \end{proposition}
  
  \begin{proof}
  Our assumptions ensure that $\zeta$ satisfies a central limit theorem conditionally to $Y$, with a limit, $\mathcal N(0, \Gamma)$ that is independent of the value of $U$. This implies that the limit is also unconditional.
  
%
  \end{proof}

  We are now able to present our principal result which provides a high-probability upper bound for $J$.

  \begin{theorem}
  \label{th:J}
  Let $Z \sim \mathcal{N}(0,1)$ be a randomization variable, independent of $(T_g, g\in G)$. For $i\in \{1, \ldots, n\}$ and $C>0$, define
  $$H_i(C)={\left( \frac{-(t_i  Z+C)+\sqrt{{(t_i Z+C)}^2+4(1-t_i)(|G|-|\hat{D}_{t_i}|)}}{2(1-t_i)} \right)}^2.$$
  
  Then
  $$\mathbb{P}\left(J \geq |G|-\max_{i}H_i(C)\right) \overset{|G| \rightarrow \infty}{\longrightarrow} \eta(C)$$
  where $\eta(C) \leq \exp\left(-\frac{C^2}{2 t_n}\right)$.
  \end{theorem}    
As a consequence, given $\varepsilon>0$, let $C_\varepsilon=\sqrt{-2t_n\log(\varepsilon)}$. Then 
 \begin{equation}
 \label{eq:j.hat}
 \hat J_\varepsilon=|G|-\max_{i}H_i(C_\varepsilon).
 \end{equation}
 is such that $\mP(J\leq \hat J_\varepsilon) \leq \varepsilon$ (we use a randomly sampled value of $Z$).  Let us now prove this result.
 
 \begin{proof}
 We know from Proposition \ref{prop:clt} that the vector 
 $$\frac{\zeta}{\sqrt{(|G|-J)}}+\left(t_1,...,t_k \right)  Z  \overset{\mathcal{L}}{\longrightarrow} B+\left(t_1,...,t_k \right)  Z,$$
 where $B \backsim \mathcal{N}(0,\Gamma).$ Then, since min is a continuous function on $\mathbb{R}^n$, we deduce that:
$$ \mathbb{P}\left(\min_i\left( \frac{Z_{t_i}}{\sqrt{(|G|-J)}}+t_i Z\right) < -C \right) \overset{|G| \rightarrow \infty}{\longrightarrow}  \mathbb{P}\left(\min_i (B_{i}+t_i Z )< -C \right).$$
but:
 $$ \mathbb{P}\left(\min_i \left(B_{i}+t_i Z \right)< -C \right)= \mathbb{P}\left(\max_i (-B_{i}-t_i Z )>C \right).$$
The process  $M_i=-B_{i}-t_i Z $, $i=1, \ldots, n$ is a martingale and $\exp{(\lambda M_i)}$ is a submartingale for all $\lambda\geq 0$'s. Applying Doob's inequality,  then optimizing over $\lambda$'s finally gives:
 $$\mathbb{P}\left(\min_i (B_{i}+t_i Z )< -C \right) \leq \exp\left(-\frac{C^2}{2 t_n}\right).$$

It remains to prove that  
$$\mathbb{P}\left(J \geq |G|-\max_{i}H_i(C)\right)=  \mathbb{P}\left(\min_i\left( \frac{Z_{t_i}}{\sqrt{(|G|-J)}}+t_i Z\right) < -C \right)$$
  But:
  $$  \frac{Z_{t_i}}{\sqrt{(|G|-J)}}=\frac{|\hat{D}_{t_i}|-(|G|-J)t_i-J}{\sqrt{|G|-J}}.$$  
  Then:
  \begin{multline*}
 \mathbb{P}\left(\min_i\left( \frac{Z_{t_i}}{\sqrt{(|G|-J)}}+t_i Z\right) < -C \right)\\
 =\mathbb{P}\left( \exists \ i  \text{ s.t  } \frac{|\hat{D}_{t_i}|-(|G|-J)t_i-J}{\sqrt{|G|-J}}+t_i Z <-C \right).
 \end{multline*}
  Solving the quadratic inequality for $\sqrt{|G|-J} $, one finds that $\frac{|\hat{D}_{t_i}|-(|G|-J)t_i-J}{\sqrt{|G|-J}}+t_i Z <-C$ is equivalent to $\sqrt{|G|-J} <\sqrt{H_i(C)}$, which completes the proof.
  \end{proof} 
  
 \subsection{Application to the coarse-to-fine algorithm}
 
 The previous section provided us with an estimator $\hat J = \hat J_\varepsilon$ in \eqref{eq:j.hat} such that $\hat J> J$ with probability larger than $1-\varepsilon$, which implies that
 $$\FWER(\hat A) \leq |V|\, p_{00}(\tg, \tv) +\nu_G \, \hat J\,  p_0(\tv),$$
 with probability $1-\varepsilon$ at least. 
 
We previously chose constants $\tg$ and $\tv$ by optimizing the detection rate on a well-chosen alternative hypothesis subject to the upper-bound being less than a significance level $\alpha$. This was done using a deterministic upper-bound of $J$, but cannot be directly applied with a data-based estimation of $J$ since this would yield data-dependent constant $\tg$ and $\tv$, which cannot be plugged into the definition of the set $\hat A$ without invalidating our estimation of the FWER.
In other terms, if, for a fixed number $J'$, one defines $\hat A_{J'}$ to be the discovery set obtained by optimizing $\tg$ and $\tv$ subject to $|V|\, p_{00}(\tg, \tv) +\nu_G \, J'\,  p_0(\tv) \leq \alpha $, our previous results imply that $\FWER(\hat A_{J'}) \leq \alpha$ for all $J'\geq J$, but not necessarily that $\FWER(\hat A_{\hat J}) \leq \alpha + \varepsilon$.
 
 A simple way to address this issue is to replace $\hat A_{\hat J}$ with
 \[
 \tilde A = \bigcap_{J'\leq \hat J}\hat A_{J'}.
 \]
 Because $\tilde A \subset \hat A_J$ with probability at least $1-\varepsilon$, we have
 \[
 \FWER(\tilde A) = P(\tilde A\cap V_0 \neq \emptyset) \leq P(\hat A_J\cap V_0 \neq \emptyset) + \varepsilon = \FWER(\hat A_J) + \varepsilon,
 \]
 so that $\tilde A$ controls the FWER at level $\alpha+\varepsilon$ as intended.
 
\subsection{Justification of A1 and A2}       
\label{sec:justify}
  
We check that conditions A1 and A2 are satisfied for the two situations that we consider in this paper.
In the example from section \ref{sec:3}, we can take (using the same notation and introducing the c.d.f. of a chi-square distribution)
\[
T_g  = 1-F_{\chi^2(\nu_G-1)}\left(\frac{{\lVert P_g\mathbf Y \rVert}^2-{\lVert \bar{\mathbf Y} \rVert}^2}{\sigma_{\mathbf Y}^2}\right).
\]
(Recall that $P_g$ is the orthogonal projection on the space generated by $\mathbf X_g = (\mathbf X_v, v\in g)$ and $\mathbf 1$. We also let $\sigma_{\mathbf Y}^2$ be the empirical variance of $\mathbf Y$.)

Note that the conditional distribution of $T_g$ given $\mathbf Y=\mathbf y$ is always uniform over $[0,1]$ and therefore does not depend on $\mathbf y$, which proves that $T_g$ and $\mathbf Y$ are independent. Similarly, taking $g_1\neq g_2 \in G_0$, $T_{g_1}$ and $T_{g_2}$ are conditionally independent given $\mathbf Y$ (because $\mathbf X_{g_1}$ and $\mathbf X_{g_2}$ are independent). But $\mathbf T_{g_1}$ and $\mathbf T_{g_2}$  being conditionally independent given $\mathbf Y$ and each of them independent of $\mathbf Y$ implies that the three variables are mutually independent.\\

The same argument can be applied to the non-parametric case, when (now using notation from that section) one assumes that scores are such that $\rho_g(\mathbf U) = \rho_g(\mathbf Y, \mathbf X_g)$, and uses --to simplify the discussion-- the statistic
$$T_{g}=\mu\left(\xi: \rho_g(\mathbf Y, \mathbf X_g)\leq \rho_g(\xi\bullet (\mathbf Y, \mathbf X_g))\right),$$
assuming, in addition, the following. If we denote by $\mathcal{X}_g$ the space where the random variable $\mathbf{X}_g$ takes its values, There exists a group $\tilde{\S}$, a group  isomorphism $\phi$ between $\S$ and $\tilde{\S}$ and a group action   of $\tilde{\S}$ on $\mathcal{X}_g$ that we will denote by $\tilde{\bullet}$ satisfying the two following conditions:

\begin{itemize}
\item The distribution of $\mathbf{X}_g$ is invariant under the action $\tilde{\bullet}$.
\item $\rho_g(\xi \bullet (\mathbf{Y}, \mathbf{X}_g))=\rho_g(\mathbf{Y},\phi(\xi) \tilde{\bullet} \mathbf{X}_g) $
\end{itemize}

For example, for  permutation tests, the group $\tilde{\S}$ is simply the group of permutations $\S$ itself. The isomorphism $\phi$ is the inverse map $\xi \rightarrow \xi^{-1}$. The group action  $\tilde{\bullet}$ is just the permutation of the observations. Finally, $\rho_g$ can be any score that is symmetric with respect to the observations. \\
Assuming these conditions, one  can immediately apply lemma $\ref{lem:basic}$ to  conclude.

\section{Discussion}

Given a partition of the space of hypotheses, the basic assumption
which allows the coarse-to-fine multiple testing algorithm to obtain greater
power than the Bonferroni-Holm approach at the same FWER level is that
the distribution of the numbers of active hypotheses across the cells
of the partition is non-uniform.  The gap in performance is then
roughly proportional to the degree of skewness.  The test derived for
the parametric model can be seen as a generalization to coarse-to-fine
testing of the F-test for determining whether a set of coefficients is
zero in a regression model; the testing procedure derived for the
non-parametric case is a generalization of permutation tests to a
multi-level multiple testing.

This scenario was motivated by the situation encountered in
genome-wide association studies, where the hypotheses are associated
with genetic variations (e.g., SNPs), each having a location along the
genome, and the cells are associated with genes.  In principle, our
coarse-to-fine procedure will then detect more active variants to the extent that
these variants cluster in genes.  Of course this extent will depend in
practice on many factors, including effect sizes, the representation
of the genotype (i.e., the choice of variants to explore) as well as
the phenotype, and complex interactions within the genotype.  It may
be very difficult and uncommon to know anything specific about the
expected nature of the combinatorics between genes and variants.  In
some sense, ``the proof is in the pudding,'' in that one can simply
try both the standard and coarse-to-fine approaches and compare the sets of
variants detected.  Given tight control of the FWER, everything found
is likely to be real. Indeed, the analytical bounds obtained here make
this comparison possible, at least under linear model commonly used in
GWAS and in a general non-parametric model under invariance
assumptions.

Looking ahead, we have only analyzed the coarse-to-fine approach for the simplest
case of two-levels and a true partition, i.e., non-overlapping cells.
The methods for controlling the FWER for both the parametric and
non-parametric cases generalize naturally to multiple levels assuming
nested partitions.  The analytical challenge is to generalize the coarse-to-fine
approach to overlapping cells, even for two levels: while our methods for controlling the FWER remain valid, they are likely to become overly conservative if cell overlap.  This
case is of particular interest in applications, where genes are
grouped into overlapping ``pathways.''  For example, in ``systems
biology,'' cellular phenotypes, especially complex diseases such as
cancer, are studied in the context of these pathways and mutated genes
and other abnormalities are in fact known to cluster in pathways;
indeed, this is the justification for a pathway-based analysis.  Hence
the clustering properties may be stronger for variants or genes in
pathways than for variants in genes.

\bibliographystyle{plain}

\end{document}